\documentclass{amsart}

\usepackage{natbib}

\usepackage[utf8]{inputenc}
\usepackage[T1]{fontenc}
\usepackage{hyperref}       
\usepackage{url}
\usepackage{booktabs}
\usepackage{amsfonts}
\usepackage{nicefrac}
\usepackage{microtype}
\usepackage{xcolor}
\usepackage[text={440pt,575pt},headheight=9pt,centering]{geometry}

\usepackage{graphicx}
\usepackage{amsmath,amssymb,amsthm,mathtools}
\usepackage{mathrsfs}
\usepackage{comment}
\usepackage{algorithm}
\usepackage{algpseudocode}
\usepackage{tikz}
\usepackage[capitalize]{cleveref}  

\hypersetup{colorlinks=true,linkcolor=blue,citecolor=blue,urlcolor=cyan}

\newcommand{\RR}{\ensuremath{\mathbb{R}}}

\newcommand{\one}{\ensuremath{\mathbf{1}}}
\newcommand{\zero}{\ensuremath{\mathbf{0}}}

\newcommand{\diag}{\ensuremath{\mathrm{diag}}}
\newcommand{\tr}{\ensuremath{\operatorname{tr}}}
\newcommand{\supp}{\ensuremath{\operatorname{supp}}}
\newcommand{\MTPtwo}{\ensuremath{\mathrm{MTP}_2}}

\newcommand{\PSD}{\ensuremath{\succeq}}
\newcommand{\NSD}{\ensuremath{\preceq}}

\usepackage{aliascnt}

\newtheorem{theorem}{Theorem}[section]

\newaliascnt{proposition}{theorem}
\newtheorem{proposition}[proposition]{Proposition}
\aliascntresetthe{proposition}

\newaliascnt{lemma}{theorem}

\aliascntresetthe{lemma}

\newaliascnt{corollary}{theorem}
\newtheorem{corollary}[corollary]{Corollary}
\aliascntresetthe{corollary}

\newaliascnt{definition}{theorem}

\aliascntresetthe{definition}

\newaliascnt{example}{theorem}

\aliascntresetthe{example}

\newaliascnt{remark}{theorem}
\newtheorem{remark}[remark]{Remark}
\aliascntresetthe{remark}

\newaliascnt{conjecture}{theorem}

\aliascntresetthe{conjecture}

\crefname{theorem}{theorem}{theorems}
\Crefname{theorem}{Theorem}{Theorems}
\crefname{lemma}{lemma}{lemmas}
\Crefname{lemma}{Lemma}{Lemmas}
\crefname{proposition}{proposition}{propositions}
\Crefname{proposition}{Proposition}{Propositions}
\crefname{definition}{definition}{definitions}
\Crefname{definition}{Definition}{Definitions}
\crefname{remark}{remark}{remarks}
\Crefname{remark}{Remark}{Remarks}
\crefname{corollary}{corollary}{corollaries}
\Crefname{corollary}{Corollary}{Corollaries}
\crefname{example}{example}{examples}
\Crefname{example}{Example}{Examples}
\crefname{conjecture}{conjecture}{conjectures}
\Crefname{conjecture}{Conjecture}{Conjectures}

\title[]{Learning Gaussian Graphical Models under Total Positivity via Spectral Graph Sparsification}

\author{Ignacio Echave-Sustaeta Rodr\'iguez}
    \address{Department of Mathematics and Computer Science, Eindhoven University of Technology, Eindhoven, The Netherlands}
\author{Aida Abiad}
\address{Department of Mathematics and Computer Science, Eindhoven University of Technology, Eindhoven, The Netherlands}
    \author{Frank R\"ottger}
    \address{Department of Applied Mathematics, University of Twente, Enschede, The Netherlands}
\begin{document}

\begin{abstract}

Many practical data analysis tasks reduce to learning, from observed samples, how a collection of variables depend on each other. 
A widely used approach is to fit a Gaussian graphical model, which represents the dependence structure as a graph connecting the variables. 
In a number of important applications, such as financial returns, gene co-expression, and climate or network analysis, the dependencies tend to be \emph{positive}: variables move together rather than offset each other. 
Encoding this positivity through the constraint of multivariate total positivity of order two (\MTPtwo{}) yields an attractive estimator that produces accurate fits with no tuning required. 
The resulting graphs are, however, typically much denser than the underlying ground-truth model, which makes them hard to interpret and slow to use in any downstream task that operates on the graph.  
In this work, we propose a novel highly-scalable approach for learning Gaussian graphical models from data using spectral sparsification; we call it \emph{Spectral-\MTPtwo{}}.
Spectral graph sparsification is a fundamental method which aims to preserve meaningful properties of a dense graph with a sparser subgraph.
We theoretically and empirically investigate and validate our method, and show that learning Gaussian Graphical Models under \MTPtwo{} using spectral sparsification preserves \MTPtwo{} and approximates well the original model in terms of Kullback--Leibler divergence and Gaussian log-likelihood.
In simulations and applications to equity returns and gene expression, we observe that Spectral-\MTPtwo{} retains most of the fit quality of the denser \MTPtwo{} baseline, while producing substantially sparser and more interpretable graphs.
\end{abstract}
\maketitle

\section{Introduction}
Gaussian graphical models are a widely used framework for representing multivariate dependence as a graph,
with edges connecting pairs of variables that interact directly after accounting for the others \citep{lauritzen1996}. In many real-world settings, these dependencies tend to be \emph{positive}: stock returns within a sector rise and fall together with the broader market \citep{AgrawalRoyUhler2022}, co-expressed genes in a biological pathway activate jointly, and neighboring weather measurements are positively correlated. 
Encoding this positivity through the constraint of multivariate total positivity of order two (\MTPtwo{}) yields an attractive estimator: the corresponding maximum-likelihood fit captures positive interactions accurately, requires no tuning parameter, and is supported by a rich and growing theoretical literature \citep{Fallat2017,LauritzenUhlerZwiernik2019,LauritzenUhlerZwiernik2021}. The graphs returned by \MTPtwo{} estimation are, however, typically much denser than the underlying ground-truth, retaining spurious edges that obscure the true network. This limits both interpretability for domain experts and the tractability of downstream tasks (visualization, clustering, graph-based prediction), since many graph-level computations are efficient on sparse inputs but intractable in general.

Spectral graph theory has revolutionized the algorithmic treatment of large and dense graphs over the last two decades. 
The spectral machinery is considered as very strong for applications in theoretical computer science; an instance of it is the work of \citet{ST2004} on spectral sparsification of Laplacian solvers. 
\emph{Graph spectral sparsification} is the approximation of an arbitrary graph by a sparser graph while retaining its essential spectral properties \citep{SpSr2011,BatsonSpielmanSrivastava2012}.
Such approximations have been used to build nearly-linear-time Laplacian solvers \citep{KLP12} and have found applications in graph-based learning, clustering, and dimensionality reduction \citep{BravoHermsdorffGunderson2019,CalandrielloLazaricKoutisValko2018}, and in algorithmic statistics and machine learning more broadly \citep{ChengChengLiuPengTeng2015,WangZhaoFeng2022}. 
To the best of our knowledge, and despite its potential, this spectral machinery has not yet been used in the context of Gaussian graphical model estimation; this is the gap we cover in this work.

In this paper, we propose to use spectral sparsification as a post-estimation step for Gaussian \MTPtwo{} graphical models: we fit a (possibly dense) \MTPtwo{} estimate, sparsify the underlying graph, and re-estimate \MTPtwo{} on the sparser support. We perform the procedure with the deterministic linear-size sparsifier of \citet{BatsonSpielmanSrivastava2012} (BSS), which gives precise control over the edge count. 
We call this estimator \emph{Spectral-\MTPtwo{}}, and show that it preserves the \MTPtwo{} property and approximates its input under standard distributional and information-theoretic measures. 
Simulations and applications to equity returns and gene expression illustrate that Spectral-\MTPtwo{} retains most of the fit quality of the dense \MTPtwo{} baseline while producing substantially sparser graphs that are easier to interpret and to use downstream.

\paragraph{Related work.}
Spectral sparsification was introduced by \citet{ST2004} and given a near-optimal form by \citet{SpSr2011,BatsonSpielmanSrivastava2012}, with subsequent refinements through faster randomized constructions \citep{KLP12,FungHKPR2011} and direct sparsification of strictly diagonally dominant $M$-matrices (SDDM) via approximate Gaussian elimination \citep{KyngSachdeva2016}. 
In statistics and machine learning, this powerful spectral machinery has been used to build nearly-linear-time samplers for Gaussian graphical models with SDDM precision matrices \citep{ChengChengLiuPengTeng2015}, to sparsify graphs prior to Laplacian-regularized regression while preserving statistical validity \citep{SadhanalaWangTibshirani2016}, and to combine ridge spectral sparsification with Laplacian learning for scalable semi-supervised classification \citep{CalandrielloLazaricKoutisValko2018}. 
The closest precedent to our work is that of \citet{WangZhaoFeng2022}, who propose a scalable algorithm for \emph{Laplacian-constrained} Gaussian graphical models that interleaves spectral sparsification and densification to identify the most spectrally critical edges. 
Our setting is strictly more general: an \MTPtwo{} precision matrix is not necessarily a Laplacian nor an SDDM matrix, and we act as a post-estimation step with explicit Loewner-order, log-likelihood, and KL guarantees. On the modeling side, Gaussian \MTPtwo{} graphical models have been studied through $\ell_1$-penalized likelihood \citep{FHT2007,SlawskiHein2015,WangRoyUhler2020,CaiCardosoPalomarYing2023}, the GOLAZO penalty framework of \citet{LZ2022}, adaptive multi-stage estimators \citep{YingCardosoPalomar2023}, and scalable bridge-block decomposition algorithms \citep{WangYingPalomar2023}.

\section{Background}\label{sec:preliminaries}

\subsection{Gaussian graphical models and \MTPtwo{}}\label{sec:gaussian_prelims}
Let $X=(X_1,\ldots,X_d)$ be a Gaussian random vector with positive definite covariance matrix $\Sigma$ and precision matrix $K:=\Sigma^{-1}$. 
We call
\[
\ell(K;S) := \log\det K - \tr(KS)
\]
the normalized Gaussian log-likelihood of $K$ given a sample covariance $S$ obtained from i.i.d.\ samples.
For a simple graph $G=(V,E)$ on $V=\{1,\ldots,d\}$, the Gaussian graphical model associated with $G$ consists of all such distributions that satisfy $K_{ij}=0$ whenever $\{i,j\}\notin E$ \citep{lauritzen1996}.
A distribution on $\RR^d$ with density $f$ is \emph{multivariate totally positive of order two} (\MTPtwo{}) if
\begin{equation}\label{eq:mtp2_density}
f(x)\,f(y) \;\le\; f(x\vee y)\,f(x\wedge y) \qquad \text{for all } x,y\in\RR^d,
\end{equation}
with $\vee$ and $\wedge$ the componentwise maximum and minimum. Condition~\eqref{eq:mtp2_density} is a strong form of positive dependence: it implies nonnegative correlations between any two coordinates and continues to hold after conditioning on any subset of the remaining coordinates. 
For Gaussians it is equivalent to $K$ being a symmetric positive definite $M$-matrix, that is
\begin{equation}\label{eq:mtp2_gaussian}
K_{ij}\le 0 \text{ for all } i\neq j,
\end{equation}
so that all partial correlations are nonnegative~\citep{LauritzenUhlerZwiernik2019,WangRoyUhler2020}.

\subsection{Graphs, Laplacians, and spectral sparsification}\label{sec:graph_prelims}
We consider weighted simple graphs $G=(V,E,c)$ with vertex set $V$, edge set $E$, and symmetric weight function $c:V\times V\to [0,\infty)$. We write $d=|V|$ and $m=|E|$. The weighted adjacency matrix is $A=(c_{ij})$, and the degree of vertex $i$ is
$k_i := \sum_{j=1}^d c_{ij}.$
The Laplacian matrix of $G$ is $L_G := D-A$, where $D=\diag(k_1,\ldots,k_d)$ is the diagonal degree matrix. For a connected weighted graph, $L_G$ is positive semidefinite, satisfies $L_G\one=0$, and has rank $d-1$.

To compare such matrices we use the \emph{Loewner partial order} on the cone of symmetric matrices: for symmetric $A,B\in\RR^{d\times d}$ we write $A\NSD B$ (equivalently $B\PSD A$) when $B-A$ is positive semidefinite, that is, $z^\top(B-A)z\ge 0$ for every $z\in\RR^d$, and reserve $A\succ 0$ for strict positive definiteness. When $A,B\succ 0$, a sandwich $A\NSD B\NSD\kappa A$ for some $\kappa\ge 1$ is equivalent to the multiplicative two-sided bound $z^\top Az\le z^\top Bz\le\kappa\,z^\top Az$ for all $z\in\RR^d$, and also to all eigenvalues of $A^{-1/2}BA^{-1/2}$ lying in $[1,\kappa]$. Graph Laplacians annihilate $\one$, so for $L_G,L_{G'}$ a Laplacian comparison is naturally restricted to the orthogonal complement $\one^\perp$.

A standard notion of approximation for weighted graphs uses this order. For $\varepsilon\in(0,1)$, we say that a weighted graph $G'$ on the same vertex set is an $\varepsilon$-spectral approximation of $G$ if
\begin{equation}\label{eq:spectral_approx_graph}
(1-\varepsilon) z^\top L_G z \le z^\top L_{G'} z \le (1+\varepsilon) z^\top L_G z
\qquad \text{for all } z\in\RR^d,
\end{equation}
equivalently $(1-\varepsilon)L_G \NSD L_{G'} \NSD (1+\varepsilon)L_G$ on $\one^\perp$. For small $\varepsilon$, \eqref{eq:spectral_approx_graph} preserves the nonzero Laplacian eigenvalues up to a factor of $1\pm\varepsilon$ and gives analogous bounds for normalized Laplacian eigenvalues, effective resistances, and the associated Gaussian quadratic forms \citep{SpSr2011}.

Spectral sparsification was introduced by \citet{ST2004}, who showed that every weighted graph admits a sparse spectral approximation obtained in nearly-linear time with a polylogarithmic edge count. The randomized leverage-score sampling construction of \citet{SpSr2011} tightens this to $O(\varepsilon^{-2}d\log d)$ edges with the spectral approximation succeeding with probability at least $1/2$. Subsequent work has developed further nearly-linear-time randomized constructions \citep{KLP12,FungHKPR2011} and sparsification via approximate Gaussian elimination on SDDM matrices \citep{KyngSachdeva2016}. Our method uses the deterministic linear-size construction of \citet{BatsonSpielmanSrivastava2012}, stated below.
\begin{theorem}[Twice-Ramanujan spectral sparsification, \citealp{BatsonSpielmanSrivastava2012}]\label{thm:bss}
Let $L$ be a weighted graph Laplacian on $d$ vertices and fix $\eta>1$. There is a deterministic polynomial-time algorithm that outputs a reweighted subgraph Laplacian $\widetilde L$ supported on at most $\lceil\eta(d-1)\rceil$ edges and satisfying
\[
L \;\NSD\; \widetilde L \;\NSD\; \kappa(\eta)\,L \quad \text{on } \one^\perp,
\qquad
\kappa(\eta) = \Bigl(\tfrac{\sqrt\eta+1}{\sqrt\eta-1}\Bigr)^2.
\]
Equivalently, after rescaling $\widetilde L$ by a constant, $\widetilde L$ is a $(1\pm\varepsilon)$-spectral approximation of $L$ in the sense of \eqref{eq:spectral_approx_graph} with $\varepsilon=(\kappa(\eta)-1)/(\kappa(\eta)+1)$ and $O(d/\varepsilon^2)$ nonzero off-diagonal entries.
\end{theorem}

The construction is iterative: \citet{BatsonSpielmanSrivastava2012} build $\widetilde L$ over $q=\lceil\eta(d-1)\rceil$ rank-one updates of a partial sum, controlling two barrier potentials on its spectrum at each step. We defer the precise mechanics, together with the per-iteration cost analysis, to \Cref{sec:bss_details}.

\section{Spectral sparsification of Gaussian \MTPtwo{} precision matrices}\label{sec:gaussian}

Our objective is to sparsify the precision matrix of an \MTPtwo{} Gaussian distribution. 
Clearly, \Cref{thm:bss} is not directly applicable, as the precision matrix of a non-degenerate Gaussian distribution cannot be a Laplacian matrix.
Our construction routes the sparsification through an auxiliary Laplacian obtained by rescaling, and then transports the sparse approximation back to the original coordinates. We describe our approach for a generic estimated \MTPtwo{} precision matrix $\widehat K$. In our experiments $\widehat K$ is the unpenalized maximum likelihood estimator (MLE) under \MTPtwo{}, that is
\begin{equation}
    \widehat K := \arg\max\bigl\{\ell(K;S): K\text{ is \MTPtwo{}}\bigr\}, \label{eq:MTP2-MLE}
\end{equation}
but nothing below depends on that choice.

The first step turns $\widehat K$ into an SDDM matrix. We pick a positive diagonal matrix $\Xi$ and form $B := \Xi\widehat K\Xi$, choosing $\Xi$ so that the row-sum vector $u := B\one$ is strictly positive. The existence of such a $\Xi$ is guaranteed by \Cref{thm:positive_row_sum_scaling} below, and a concrete linearly convergent fixed-point iteration is given in \Cref{rem:sinkhorn_iteration}. The rescaled matrix $B$ inherits symmetry, positive definiteness, and the nonpositive off-diagonal sign pattern of $\widehat K$, and the positive row-sum vector $u$ promotes it to an SDDM matrix.

\begin{theorem}[Positive row-sum scaling]\label{thm:positive_row_sum_scaling}
Let $W\in\RR^{d\times d}$ be symmetric and positive definite, and let $u\in\RR_{>0}^d$. Then there exists a positive diagonal matrix $\Xi=\diag(\xi_1,\ldots,\xi_d)$ such that $\Xi W\Xi\one=u$.
\end{theorem}

Writing $D := \diag(u)$, the SDDM matrix splits naturally as $B = L_B + D$, where $L_B := B - D$ is a graph Laplacian whose edge set coincides with the sparsity pattern of $\widehat K$. We apply \Cref{thm:bss} to $L_B$ to obtain a sparse Laplacian $\widetilde L_B$, reassemble $\widetilde B := \widetilde L_B + D$, and undo the rescaling to define the sparsified estimator
\begin{equation}
    \widetilde K := \Xi^{-1}\widetilde B\,\Xi^{-1}.\label{eq:sparsified}
\end{equation}
Each step preserves positive definiteness and the nonpositive off-diagonal sign pattern, and carries its own Loewner-order guarantee. Stepwise statements and proofs are collected in \Cref{sec:proofs} (\Cref{prop:B_is_SDDM,prop:LB_is_laplacian,prop:decomp_approx}). Composing them gives the main spectral and distributional guarantees.

\begin{theorem}\label{thm:gaussian_pullback}
Let $\widehat K$ be an \MTPtwo{} precision matrix, $\varepsilon\in(0,1)$, and $\widetilde K$ obtained by the construction above with $\widetilde L_B$ an $\varepsilon$-spectral approximation of $L_B$. Then $\widetilde K$ is a symmetric positive definite $M$-matrix and satisfies
$(1-\varepsilon)\widehat K\NSD\widetilde K\NSD(1+\varepsilon)\widehat K$.
Moreover, since the final rescaling multiplies each off-diagonal entry of $\widetilde B$ by the positive factor $\xi_i^{-1}\xi_j^{-1}$, the matrices $\widetilde K$ and $\widetilde B$ have exactly the same zero pattern, so the sparsification carried out on $\widetilde L_B$ is preserved in $\widetilde K$.
\end{theorem}

\begin{corollary}\label{cor:mtp2_preserved_kl}
Let $\widehat K$ and $\widetilde K$ be as in \Cref{thm:gaussian_pullback}. Then the centered Gaussian with precision matrix $\widetilde K$ is \MTPtwo{}. If in addition $\varepsilon\le\tfrac12$, then
$D_{\mathrm{KL}}\!\bigl(\mathcal{N}(0,\widehat K^{-1})\,\big\|\,\mathcal{N}(0,\widetilde K^{-1})\bigr)\le d\varepsilon^2/2$.
\end{corollary}

\Cref{alg:bss_refit} collects the resulting pipeline. The dominant cost is the deterministic BSS step in line 4, which runs in polynomial time via the barrier-function construction of \citet{BatsonSpielmanSrivastava2012} and produces a $(1\pm\varepsilon)$-spectral approximation of $L_B$ with $O(d/\varepsilon^2)$ off-diagonal entries (\Cref{thm:bss}).

\begin{algorithm}[ht]
\caption{Spectral-\MTPtwo{}: BSS sparsification of an \MTPtwo{} precision matrix}
\label{alg:bss_refit}
\begin{algorithmic}[1]
\Require Estimated \MTPtwo{} precision matrix $\widehat K\in\RR^{d\times d}$; BSS parameter $\eta>1$; covariance $S\in\RR^{d\times d}$ for the optional MLE refit
\State \textbf{SDDM scaling:} compute a positive diagonal $\Xi$ with $\Xi\widehat K\Xi\one=u>0$ (\Cref{rem:sinkhorn_iteration}); set $B\gets\Xi\widehat K\Xi$
\State \textbf{Decomposition:} $L_B\gets B-\diag(u)$ (graph Laplacian)
\State \textbf{BSS sparsification:} run the deterministic algorithm of \Cref{thm:bss} on $L_B$ at parameter $\eta$ to obtain $\widetilde L_B$ with $\lceil\eta(d-1)\rceil$ edges and $L_B\NSD\widetilde L_B\NSD\kappa(\eta)L_B$ on $\one^\perp$; rescale to recover the symmetric $(1\pm\varepsilon)$ form with $\varepsilon=(\kappa(\eta)-1)/(\kappa(\eta)+1)$
\State \textbf{Reassemble:} $\widetilde B\gets\widetilde L_B+\diag(u)$
\State \textbf{Undo the rescaling:} set $\widetilde K\gets\Xi^{-1}\widetilde B\,\Xi^{-1}$
\State \textbf{MLE refit (optional):} \Return $\widetilde K^{\mathrm{MLE}}\gets\arg\max\{\ell(K;S):K\text{ is \MTPtwo{}},\ \supp(K)\subseteq\supp(\widetilde K)\}$
\end{algorithmic}
\end{algorithm}

\subsection{Perturbation of the Gaussian log-likelihood}\label{sec:gaussian_likelihood}

The Loewner-order guarantee of \Cref{thm:gaussian_pullback} directly controls the Gaussian log-likelihood. 
Let $T\in\RR^{d\times d}$ be a positive semidefinite matrix that may be different from the (training) sample covariance~$S$ used for estimating $\widehat K$ as in \eqref{eq:MTP2-MLE} and $\widetilde K$ in \eqref{eq:sparsified}.
For example, $T$ could be the sample covariance of a test data set in a validation setup.
The next theorem expresses the gap $\ell(\widetilde K;T)-\ell(\widehat K;T)$ as a Kullback--Leibler divergence plus a trace term that vanishes when $T$ equals the ground-truth covariance $\widehat K^{-1}$.

\begin{theorem}[Bregman identity for the log-likelihood gap]\label{thm:gaussian_loglik}
Let $\widehat K$ and $\widetilde K$ be as in \Cref{thm:gaussian_pullback}. For every symmetric positive semidefinite $T\in\RR^{d\times d}$,
\begin{equation}\label{eq:loglik_identity}
\ell(\widetilde K;T) - \ell(\widehat K;T)
\;=\;
-2\,D_{\mathrm{KL}}\!\bigl(\mathcal{N}(0,\widehat K^{-1})\,\big\|\,\mathcal{N}(0,\widetilde K^{-1})\bigr)
\;-\;\tr\!\bigl((\widetilde K-\widehat K)(T-\widehat K^{-1})\bigr).
\end{equation}
\end{theorem}

The two terms on the right-hand side play distinct roles. The first is \emph{non-positive}: any sparsifier $\widetilde K\neq\widehat K$ pays a structural cost equal to the KL divergence between the two implied Gaussians, and never gains from this term alone.

To turn the identity into a quantitative bound, we control each term in turn. The first term is non-positive: $-2D_{\mathrm{KL}}\le 0$ for any $\widetilde K$. For the second term, the Loewner bound $-\varepsilon\widehat K\NSD\widetilde K-\widehat K\NSD\varepsilon\widehat K$ together with operator--trace duality give, for every $\varepsilon\in(0,1)$,
\[
\bigl|\tr\!\bigl((\widetilde K-\widehat K)(T-\widehat K^{-1})\bigr)\bigr|
\;\le\;\varepsilon\,R(T),
\qquad
R(T):=\bigl\|\widehat K^{1/2}(T-\widehat K^{-1})\widehat K^{1/2}\bigr\|_*,
\]
where $\|\cdot\|_*$ is the trace (Schatten--1) norm. The quantity $R(T)$ measures the residual between $T$ and $\widehat K^{-1}$ after a rescaling depending on $\widehat K$. The sign structure of the identity then yields an \emph{asymmetric} two-sided bound. The upper side, which only drops the non-positive KL term, holds for every $\varepsilon\in(0,1)$:
\begin{equation}\label{eq:loglik_upper}
\ell(\widetilde K;T)-\ell(\widehat K;T)\;\le\;\varepsilon\,R(T).
\end{equation}
For the lower side we additionally use Corollary~\ref{cor:mtp2_preserved_kl}: provided $\varepsilon\le\tfrac12$, the bound $0\le 2D_{\mathrm{KL}}\le d\varepsilon^2$ controls the KL term, and combining with the bound above for the trace term we get
\begin{equation}\label{eq:loglik_bound}
-d\varepsilon^2 - \varepsilon\,R(T)\;\le\;\ell(\widetilde K;T)-\ell(\widehat K;T)\;\le\;\varepsilon\,R(T)
\qquad(\varepsilon\le\tfrac12).
\end{equation}
The sparsifier can therefore lose at most $d\varepsilon^2+\varepsilon R(T)$ but gain at most $\varepsilon R(T)$.

These bounds admit a sharper interpretation in the two natural regimes of practical interest, namely when $T=S$ is the training covariance used to fit $\widehat K$ and when $T$ is independent of $\widehat K$ (e.g.\ the sample covariance of test data). In particular, $\widetilde K$ provably cannot surpass $\widehat K$ on training data but can outperform it out of sample whenever a precise inner-product condition is satisfied. We defer this discussion to \Cref{sec:regimes}.

The same Loewner argument also yields a non-asymptotic Frobenius decomposition of the parameter error of $\widetilde K$ that translates consistency of $\widehat K$ into consistency of $\widetilde K$ at the same rate when $\varepsilon$ shrinks fast enough; we defer the precise statement to \Cref{prop:gaussian_consistency} in the appendix.

Let $E_{\mathrm{sp}}:=\bigl\{\{i,j\}:\,i,j\in[d],\,i\neq j,\,\widetilde K_{ij}\neq 0\bigr\}$ denote the edge set of $\widetilde K$. Because $\widetilde K$ is \MTPtwo{} and supported on $E_{\mathrm{sp}}$, refitting the \MTPtwo{}-MLE on this edge set provides an improvement in terms of Gaussian log-likelihood given the sample covariance $S$.

\begin{corollary}[Log-likelihood guarantee after MLE refit]\label{cor:loglik_refit}
Let $\widetilde K$ be as in \Cref{thm:gaussian_pullback} and $E_{\mathrm{sp}}$ its edge set as above. Define
\[
\widetilde K^{\mathrm{MLE}} := \arg\max\bigl\{\ell(K;S): K\text{ is \MTPtwo{}},\; K_{ij}=0\text{ for all }i\ne j\text{ with }\{i,j\}\notin E_{\mathrm{sp}}\bigr\}.
\]
Then for every $\varepsilon\in(0,1)$,
\begin{equation}\label{eq:refit_pythagoras}
\ell(\widetilde K^{\mathrm{MLE}};S) - \ell(\widetilde K;S)
\;\ge\;
2\,D_{\mathrm{KL}}\!\bigl(\mathcal{N}(0,(\widetilde K^{\mathrm{MLE}})^{-1})\,\big\|\,\mathcal{N}(0,\widetilde K^{-1})\bigr)
\;\ge\; 0,
\end{equation}
and, provided $\varepsilon\le\tfrac12$,
\[
\ell(\widetilde K^{\mathrm{MLE}};S) - \ell(\widehat K;S)
\;\ge\;
-d\varepsilon^2 - \varepsilon\,R(S)
\;+\;
2\,D_{\mathrm{KL}}\!\bigl(\mathcal{N}(0,(\widetilde K^{\mathrm{MLE}})^{-1})\,\big\|\,\mathcal{N}(0,\widetilde K^{-1})\bigr).
\]
\end{corollary}

The refit therefore not only matches but generically improves \emph{strictly} on $\widetilde K$, by an amount at least equal to the KL divergence from the refit-implied to the BSS-implied distribution. Using the sparsified subgraph followed by an \MTPtwo{}-MLE refit is therefore the natural practical recipe.

\section{Simulation study}\label{sec:simulations}

Throughout this section we operate within the positive-dependence Gaussian regime in which our methodology is well posed: the population precision matrix is required to be a symmetric positive definite $M$-matrix (\Cref{eq:mtp2_gaussian}), so that all partial correlations are nonnegative. The simulations illustrate the trade-off that spectral sparsification is meant to address: the \MTPtwo{}-MLE delivers a high-quality likelihood fit but tends to over-connect the graph (including many edges in the graph that are not in the ground-truth), while $\ell_1$-penalized methods such as the Graphical Lasso are sparser but sacrifice fit quality. Spectral-\MTPtwo{} starts from the \MTPtwo{} fit and removes the least structurally important edges, aiming to retain as much of the log-likelihood as possible while drastically reducing the edge count.

We compare three estimators. The first is the Graphical Lasso (GLasso) of \citet{FHT2007}, which serves as our $\ell_1$-penalized baseline. The second is the unpenalized \MTPtwo{}-MLE \citep{LauritzenUhlerZwiernik2019,LZ2022}, which serves as our positively dependent baseline. The third is Spectral-\MTPtwo{} (\Cref{alg:bss_refit}), which runs the deterministic BSS sparsifier on the auxiliary Laplacian $L_B$ at parameter $\eta$ and refits the \MTPtwo{}-MLE on the resulting subgraph. Tuning parameters ($\lambda$ for GLasso, $\eta$ for Spectral-\MTPtwo{}) are selected by BIC on the training data. All \MTPtwo{} fits use the \texttt{golazo} R package \citep{LZ2022}, and the code is included in the supplementary material.

Throughout this section we summarize each fit by five scalar quantities.
In order to assess predictive fit we use the \emph{test log-likelihood} $\ell(\widehat K;S_{\mathrm{test}})=\log\det\widehat K-\tr(\widehat K S_{\mathrm{test}})$, evaluated on the held-out sample covariance $S_{\mathrm{test}}$ and measured in \emph{nats} (units of natural logarithm $\log_e$). The \emph{edge count} is the number of off-diagonal entries of $\widehat K$ whose absolute value exceeds a small numerical tolerance, included to discount the residuals the GOLAZO solver leaves on entries it is trying to zero out. As a result, the reported count may differ from the algorithm's intended edge set by a few entries with magnitude near this threshold. For graph recovery against the known true edge set $E^\star$, write $\mathrm{TP}$, $\mathrm{FP}$, $\mathrm{FN}$ for the numbers of estimated edges that are correct, spurious, and missed, respectively. Then \emph{precision} $P=\mathrm{TP}/(\mathrm{TP}+\mathrm{FP})$ is the fraction of estimated edges that lie in $E^\star$, \emph{recall} $R=\mathrm{TP}/(\mathrm{TP}+\mathrm{FN})$ is the fraction of true edges recovered, and the \emph{F1 score} $2PR/(P+R)$ is their harmonic mean. Higher precision means fewer spurious edges, higher recall means fewer missed edges, and F1 balances the two.

\subsection{Scale-free graph simulation}\label{sec:sim_gaussian}

The true precision matrix $K^\star$ is the Laplacian of a Barab\'asi--Albert scale-free graph on $d=150$ nodes with attachment parameter $m=2$ ($297$ true edges, heavy-tailed degrees with maximum $17$), plus a scaled identity matrix $\delta I_d$ ($\delta=0.5$) to ensure positive definiteness. The resulting $K^\star$ is an $M$-matrix and the model is \MTPtwo{}. Hubs create many triangles and the leverage distribution is highly skewed, which stresses both the randomized and the deterministic sparsifiers. We draw $n_{\mathrm{train}}=400$ training and $n_{\mathrm{test}}=2{,}000$ test observations over $B=15$ replications, and evaluate each method by test log-likelihood, edge count, and F1 against the true edge set.

\Cref{tab:sim_large} reports the comparison. The \MTPtwo{}-MLE recovers a graph that is much denser than the truth (roughly five times the true edge count), with high recall but very low precision. GLasso loses substantial log-likelihood relative to the \MTPtwo{}-MLE.
Spectral-\MTPtwo{} attains the highest test log-likelihood of the three methods while using less than a quarter of the \MTPtwo{}-MLE's edges, and BIC selects an interior value of $\eta$ on the path.

\begin{table}[ht]
\centering
\footnotesize
\begin{tabular}{lrrrrr}
\toprule
Method & Edges & Test log-lik.\ & Precision & Recall & F1 \\
\midrule
GLasso                   & $409.4_{(8.7)}$    & $21{,}803_{(131)}$ & $0.405_{(0.006)}$ & $0.556_{(0.005)}$ & $0.468_{(0.004)}$ \\
\MTPtwo{}                & $1{,}436.8_{(5.9)}$ & $27{,}622_{(135)}$ & $0.198_{(0.001)}$ & $\mathbf{0.956}_{(0.003)}$ & $0.327_{(0.001)}$ \\
Spectral-\MTPtwo{} (BIC) & $327.3_{(2.3)}$    & $\mathbf{28{,}198}_{(144)}$ & $\mathbf{0.745}_{(0.005)}$ & $0.821_{(0.004)}$ & $\mathbf{0.781}_{(0.003)}$ \\
\bottomrule
\end{tabular}
\caption{Scale-free simulation ($d=150$, Barab\'asi--Albert $m=2$, $n_{\mathrm{train}}=400$, $B=15$ replications, true edges $=297$). Means across replications with standard errors in parentheses. Bold: best on each metric.}
\label{tab:sim_large}
\end{table}

\Cref{tab:sim_large_paths} summarizes the Spectral-\MTPtwo{} path as the BSS parameter $\eta$ varies. Each row reports the corresponding spectral approximation factor $\varepsilon=(\kappa(\eta)-1)/(\kappa(\eta)+1)$ alongside the empirical edge count, test log-likelihood, and F1. Smaller $\eta$ corresponds to a sparser subgraph at a looser certificate. The BIC-selected fit (\Cref{tab:sim_large}) sits at $\sim$$327$ edges, between $\eta=2.5$ and $\eta=3$.

\begin{table}[ht]
\centering
\footnotesize
\begin{tabular}{rrrrr}
\toprule
$\eta$ & $\varepsilon$ & Edges & Test log-lik. & F1 \\
\midrule
$1.50$ & $0.98$ & $223.9_{(0.1)}$ & $27{,}890_{(124)}$ & $0.815_{(0.004)}$ \\
$1.75$ & $0.96$ & $260.3_{(0.3)}$ & $28{,}392_{(141)}$ & $\mathbf{0.827}_{(0.004)}$ \\
$2.00$ & $0.94$ & $290.5_{(0.7)}$ & $\mathbf{28{,}397}_{(149)}$ & $0.812_{(0.004)}$ \\
$2.50$ & $0.90$ & $321.1_{(1.5)}$ & $28{,}216_{(144)}$ & $0.786_{(0.003)}$ \\
$3.00$ & $0.87$ & $330.7_{(1.4)}$ & $28{,}175_{(142)}$ & $0.778_{(0.004)}$ \\
$3.50$ & $0.83$ & $339.2_{(1.0)}$ & $28{,}140_{(146)}$ & $0.770_{(0.004)}$ \\
$4.00$ & $0.80$ & $375.0_{(1.6)}$ & $28{,}012_{(140)}$ & $0.742_{(0.004)}$ \\
$5.00$ & $0.75$ & $417.8_{(2.3)}$ & $27{,}878_{(145)}$ & $0.709_{(0.005)}$ \\
$6.00$ & $0.70$ & $447.0_{(2.2)}$ & $27{,}817_{(145)}$ & $0.688_{(0.004)}$ \\
$8.00$ & $0.63$ & $502.5_{(2.6)}$ & $27{,}701_{(138)}$ & $0.649_{(0.004)}$ \\
\bottomrule
\end{tabular}
\caption{Scale-free simulation ($d=150$, $n_{\mathrm{train}}=400$, $B=15$ replications, true edges $=297$): Spectral-\MTPtwo{} path over the BSS tuning grid. Means with standard errors in parentheses. Bold: best test log-likelihood and best F1 along the path. \MTPtwo{}-MLE baseline ($1{,}437$ edges, log-lik.\ $27{,}622$) and GLasso baseline ($409$ edges, log-lik.\ $21{,}803$) as in \Cref{tab:sim_large}.}
\label{tab:sim_large_paths}
\end{table}

\subsection{Larger scale-free simulation ($d=300$)}\label{sec:sim_xlarge}

We repeat the experiment on a Barab\'asi--Albert graph with $d=300$ nodes and $m=2$ ($597$ true edges), otherwise mirroring \Cref{sec:sim_gaussian} with $n_{\mathrm{train}}=800$, $n_{\mathrm{test}}=3{,}000$, and $10$ replications. Edge weights are drawn uniformly from $[0.5,1.5]$ to break ties between edges of equal nominal weight.

\begin{table}[ht]
\centering
\footnotesize
\begin{tabular}{lrrrrr}
\toprule
Method & Edges & Test log-lik.\ & Precision & Recall & F1 \\
\midrule
GLasso                   & $405.8_{(4.1)}$       & $60{,}242_{(230)}$ & $\mathbf{0.733}_{(0.006)}$ & $0.498_{(0.003)}$ & $0.593_{(0.003)}$ \\
\MTPtwo{}                & $4{,}116.2_{(15.5)}$  & $79{,}061_{(259)}$ & $0.138_{(0.001)}$         & $\mathbf{0.953}_{(0.002)}$ & $0.241_{(0.001)}$ \\
Spectral-\MTPtwo{} (BIC) & $737.3_{(6.1)}$       & $\mathbf{80{,}825}_{(310)}$ & $0.694_{(0.006)}$ & $0.856_{(0.003)}$ & $\mathbf{0.766}_{(0.004)}$ \\
\bottomrule
\end{tabular}
\caption{Scale-free simulation at larger dimension ($d=300$, Barab\'asi--Albert $m=2$, $n_{\mathrm{train}}=800$, $n_{\mathrm{test}}=3{,}000$, 10 replications, true edges $=597$). Means across replications with standard errors in parentheses. Bold: best on each metric.}
\label{tab:sim_xlarge}
\end{table}

The qualitative pattern at $d=150$ persists in higher dimension: \MTPtwo{}-MLE recovers a graph nearly $7$ times the true edge count (recall $0.95$ but precision $0.14$), GLasso uses an order of magnitude fewer edges than \MTPtwo{} but its test log-likelihood lags by roughly $18{,}800$ nats (about $6.3$ nats per held-out observation, with $n_{\mathrm{test}}=3{,}000$), and Spectral-\MTPtwo{} preserves nearly all of the \MTPtwo{}-MLE log-likelihood at roughly $18\%$ of its edges and with the highest F1 of the three methods. As at $d=150$, BIC selects an interior point of the BSS path (here $\eta\approx 4$, $737$ edges); the path peak in test log-likelihood lies slightly earlier at $\eta=2$ ($568$ edges, log-lik.\ $82{,}059$, F1 $0.85$). \Cref{tab:sim_xlarge_paths} shows the full path.

\begin{table}[ht]
\centering
\footnotesize
\begin{tabular}{rrrrr}
\toprule
$\eta$ & $\varepsilon$ & Edges & Test log-lik. & F1 \\
\midrule
$1.5$ & $0.98$ & $458.8_{(2.5)}$  & $80{,}599_{(236)}$ & $0.832_{(0.003)}$ \\
$2.0$ & $0.94$ & $568.3_{(2.6)}$  & $\mathbf{82{,}059}_{(284)}$ & $\mathbf{0.848}_{(0.004)}$ \\
$2.5$ & $0.90$ & $602.9_{(2.5)}$  & $81{,}819_{(293)}$ & $0.833_{(0.003)}$ \\
$3.0$ & $0.87$ & $615.0_{(1.9)}$  & $81{,}759_{(310)}$ & $0.828_{(0.003)}$ \\
$4.0$ & $0.80$ & $731.8_{(3.6)}$  & $80{,}853_{(301)}$ & $0.769_{(0.004)}$ \\
$5.0$ & $0.75$ & $805.9_{(4.9)}$  & $80{,}430_{(287)}$ & $0.736_{(0.004)}$ \\
$6.0$ & $0.70$ & $893.0_{(4.6)}$  & $80{,}025_{(278)}$ & $0.700_{(0.003)}$ \\
$8.0$ & $0.63$ & $1{,}027.2_{(4.7)}$ & $79{,}507_{(261)}$ & $0.649_{(0.002)}$ \\
\bottomrule
\end{tabular}
\caption{Scale-free simulation ($d=300$, $n_{\mathrm{train}}=800$, $B=10$ replications, true edges $=597$): Spectral-\MTPtwo{} path over the BSS tuning grid. Means with standard errors in parentheses. Bold: best test log-likelihood and best F1 along the path. \MTPtwo{}-MLE baseline ($4{,}116$ edges, log-lik.\ $79{,}061$) and GLasso baseline ($406$ edges, log-lik.\ $60{,}242$) as in \Cref{tab:sim_xlarge}.}
\label{tab:sim_xlarge_paths}
\end{table}

\section{Experiments with real data}\label{sec:realdata}

We empirically evaluate and validate the estimators of \Cref{sec:simulations} on two datasets that are standard benchmarks in the graphical models literature and that cover complementary regimes: a long financial time series where $n\gg d$ (S\&P~500 tech/communications/discretionary stocks) and a cancer gene expression dataset where $n$ and $d$ are of the same order (GSE14333 colorectal cancer tumors). All fits use the same tuning grids and BIC selection rule as in the simulation. Code for reproducing the experiments is available in the supplementary material.

\subsection{S\&P~500 tech--communications--discretionary block ($d=141$)}\label{sec:sp500}

Following the stock-market benchmarks for Gaussian \MTPtwo{} models used by \citet{LauritzenUhlerZwiernik2019}, we take daily log-returns of S\&P~500 constituents in the technology, communications, and consumer-discretionary sectors over 2015--2024 ($n=2{,}514$ trading days, downloaded via \texttt{quantmod}). After dropping tickers with incomplete histories or download failures, $d=141$ stocks remain. Restricting to sectors with shared macro drivers makes the \MTPtwo{} assumption more plausible than on the full index, since common exposure to technology-cycle and consumer-demand factors induces positive dependence. We use a chronological 70/30 split ($n_{\mathrm{train}}=1{,}759$, $n_{\mathrm{test}}=755$).

\begin{table}[ht]
\centering
\begin{tabular}{lrrr}
\toprule
Method & Tuning & Edges & Test log-lik.\ \\
\midrule
Graphical Lasso        & $\lambda=10^{-5}$ & $3{,}226$ & $378{,}984$ \\
\MTPtwo{}              & NA               & $1{,}444$ & $\mathbf{379{,}352}$ \\
Spectral-\MTPtwo{} (BIC)     & $\eta=10$         & $629$     & $378{,}781$ \\
\bottomrule
\end{tabular}
\caption{S\&P~500 tech/communications/discretionary block ($d=141$, $n_{\mathrm{train}}=1{,}759$, $n_{\mathrm{test}}=755$, daily log-returns 2015--2024). BIC-selected tunings. Bold: best test log-likelihood.}
\label{tab:gaussian_sp500}
\end{table}

The \MTPtwo{}-MLE attains the best test log-likelihood, confirming that positive dependence is a good model for an intra-sector block fitted on a large training sample. Spectral-\MTPtwo{} is within $571$ nats of that peak using only $44\%$ of its edges ($629$ vs.\ $1{,}444$), with the BIC-selected $\eta=10$ in the interior of the grid (\Cref{tab:gaussian_sp500_sparse}). Graphical Lasso, in contrast, uses by far the largest edge count of the three methods ($3{,}226$, more than twice the \MTPtwo{}-MLE) and still trails the \MTPtwo{}-MLE by $368$ nats; Spectral-\MTPtwo{} matches GLasso to within $203$ nats while using less than $20\%$ of its edges.
\begin{table}[ht]
\centering
\small
\begin{tabular}{rrrr}
\toprule
$\eta$ & $\varepsilon$ & Edges & Test log-lik.\ \\
\midrule
$1.5$  & $0.98$ & $211$  & $374{,}662$ \\
$2.0$  & $0.94$ & $280$  & $376{,}434$ \\
$3.0$  & $0.87$ & $375$  & $377{,}870$ \\
$4.0$  & $0.80$ & $431$  & $378{,}121$ \\
$6.0$  & $0.69$ & $521$  & $378{,}413$ \\
$8.0$  & $0.63$ & $579$  & $378{,}763$ \\
$10.0$ & $0.58$ & $\mathbf{629}$  & $\mathbf{378{,}781}$ \\
$12.0$ & $0.53$ & $668$  & $378{,}882$ \\
$16.0$ & $0.47$ & $724$  & $378{,}944$ \\
$24.0$ & $0.39$ & $814$  & $379{,}113$ \\
$32.0$ & $0.34$ & $875$ & $379{,}135$ \\
\bottomrule
\end{tabular}
\caption{S\&P~500 tech/comm/discr ($d=141$): Spectral-\MTPtwo{} path over the BSS tuning grid, with $\varepsilon=(\kappa(\eta)-1)/(\kappa(\eta)+1)$. \MTPtwo{}-MLE baseline: $1{,}444$ edges, test log-lik.\ $379{,}352$. Bold: BIC-selected $\eta$.}
\label{tab:gaussian_sp500_sparse}
\end{table}

\subsection{GSE14333 colorectal cancer gene expression ($d=150$)}\label{sec:gse14333}

The second application uses the GSE14333 colorectal cancer expression dataset ($n=290$ primary tumor samples on an Affymetrix HG-U133 Plus~2 array, RMA-normalized log$_2$ intensities from NCBI GEO). Retaining the $d=150$ most-variable probes places the experiment in the moderate-dimensional regime $n\gtrsim d$, where the \MTPtwo{}-MLE exists but is close to ill-posed. A random 70/30 split gives $n_{\mathrm{train}}=203$ and $n_{\mathrm{test}}=87$.

\begin{table}[ht]
\centering
\begin{tabular}{lrrr}
\toprule
Method & Tuning & Edges & Test log-lik.\ \\
\midrule
Graphical Lasso        & $\lambda=0.75$ & $1{,}306$ & $\mathbf{-14{,}782}$ \\
\MTPtwo{}              & NA            & $889$     & $-14{,}803$ \\
Spectral-\MTPtwo{} (BIC)     & $\eta=2.5$     & $332$     & $-14{,}954$ \\
\bottomrule
\end{tabular}
\caption{GSE14333 colorectal cancer gene expression ($d=150$ most-variable probes, $n_{\mathrm{train}}=203$, $n_{\mathrm{test}}=87$). BIC-selected tunings. Bold: best test log-likelihood.}
\label{tab:gaussian_gse14333}
\end{table}

Unlike in the equity example, the Graphical Lasso wins narrowly on test log-likelihood, by only $21$ nats over the \MTPtwo{}-MLE, reflecting that gene-expression profiles carry some sign-mixed partial correlations between tumor-suppressor and proliferation-related pathways that the \MTPtwo{} constraint cannot reproduce. The three methods are nevertheless within $172$ nats of each other on test data. Spectral-\MTPtwo{} trades a $151$-nat log-likelihood cost relative to the \MTPtwo{}-MLE for a graph that retains only $37\%$ of \MTPtwo{}-MLE's edges ($332$ vs.\ $889$, equivalently $25\%$ of GLasso's $1{,}306$ edges; \Cref{tab:gaussian_gse14333_sparse}). Even when \MTPtwo{} is not the dominant model, Spectral-\MTPtwo{} delivers a much sparser graph at a modest log-likelihood cost.

\begin{table}[ht]
\centering
\small
\begin{tabular}{rrrr}
\toprule
$\eta$ & $\varepsilon$ & Edges & Test log-lik.\ \\
\midrule
$1.50$ & $0.98$ & $224$ & $-15{,}094$ \\
$1.75$ & $0.97$ & $256$ & $-15{,}026$ \\
$2.00$ & $0.94$ & $292$ & $-14{,}991$ \\
$2.50$ & $0.90$ & $\mathbf{332}$ & $\mathbf{-14{,}954}$ \\
$3.00$ & $0.87$ & $349$ & $-14{,}953$ \\
$3.50$ & $0.83$ & $358$ & $-14{,}942$ \\
$4.00$ & $0.80$ & $380$ & $-14{,}918$ \\
$5.00$ & $0.74$ & $412$ & $-14{,}899$ \\
$6.00$ & $0.69$ & $433$ & $-14{,}892$ \\
$8.00$ & $0.63$ & $473$ & $-14{,}864$ \\
\bottomrule
\end{tabular}
\caption{GSE14333 CRC ($d=150$): Spectral-\MTPtwo{} path over the BSS tuning grid, with $\varepsilon=(\kappa(\eta)-1)/(\kappa(\eta)+1)$. \MTPtwo{}-MLE baseline: $889$ edges, test log-lik.\ $-14{,}803$. Bold: BIC-selected $\eta$.}
\label{tab:gaussian_gse14333_sparse}
\end{table}

Together, the two experiments test Spectral-\MTPtwo{} in two complementary regimes: the large-$n$ financial regime in which the \MTPtwo{}-MLE is itself the best dense baseline, and the moderate-dimensional genomics regime in which the Graphical Lasso is the dense baseline and Spectral-\MTPtwo{} is asked to preserve an already-constrained \MTPtwo{} fit.

\section{Discussion}\label{sec:conclusion}

We have developed spectral sparsification as a post-estimation tool for Gaussian \MTPtwo{} graphical models. Starting from an \MTPtwo{} precision matrix, the proposed procedure exposes its underlying graph structure through a structural rescaling, applies the deterministic linear-size sparsifier of \citet{BatsonSpielmanSrivastava2012}, and recovers a sparser \MTPtwo{} approximation of the input. We call this newly proposed estimator Spectral-\MTPtwo{}, and provide explicit theoretical guarantees on its statistical performance: it remains \MTPtwo{}, approximates its input under standard distributional and information-theoretic measures, and inherits the consistency rate of its starting point. Across simulations on scale-free graphs and applications to equity returns and gene expression, Spectral-\MTPtwo{} retains, and in some regimes exceeds, the predictive performance of the dense \MTPtwo{} baseline at substantially lower edge counts, producing simpler and more interpretable graphs.

\paragraph{Limitations.} Some of our bounds require the spectral approximation factor to satisfy $\varepsilon\le\tfrac12$, so at higher values of $\varepsilon$ not all of our theoretical guarantees remain valid and Spectral-\MTPtwo{}'s worst-case behavior is less tightly controlled. Across our experiments BIC selects $\varepsilon$ in the range $[0.58,0.90]$, often beyond this regime, but Spectral-\MTPtwo{} still tracks, and on BA150 surpasses, the \MTPtwo{}-MLE on test data. Our consistency results also assume fixed dimension, and a scaling analysis as $d\to\infty$ is left to future work.

\bibliographystyle{abbrvnat}
\bibliography{references}

\begin{thebibliography}{23}
\providecommand{\natexlab}[1]{#1}
\providecommand{\url}[1]{\texttt{#1}}
\expandafter\ifx\csname urlstyle\endcsname\relax
  \providecommand{\doi}[1]{doi: #1}\else
  \providecommand{\doi}{doi: \begingroup \urlstyle{rm}\Url}\fi

\bibitem[Agrawal et~al.(2022)Agrawal, Roy, and Uhler]{AgrawalRoyUhler2022}
R.~Agrawal, U.~Roy, and C.~Uhler.
\newblock Covariance matrix estimation under total positivity for portfolio selection.
\newblock \emph{Journal of Financial Econometrics}, 20\penalty0 (2):\penalty0 367--389, 2022.

\bibitem[Batson et~al.(2012)Batson, Spielman, and Srivastava]{BatsonSpielmanSrivastava2012}
J.~Batson, D.~A. Spielman, and N.~Srivastava.
\newblock Twice-{R}amanujan sparsifiers.
\newblock \emph{SIAM Journal on Computing}, 41\penalty0 (6):\penalty0 1704--1721, 2012.

\bibitem[Bravo-Hermsdorff and Gunderson(2019)]{BravoHermsdorffGunderson2019}
G.~Bravo-Hermsdorff and L.~M. Gunderson.
\newblock A unifying framework for spectrum-preserving graph sparsification and coarsening.
\newblock In \emph{Advances in Neural Information Processing Systems}, volume~32. Curran Associates, Inc., 2019.

\bibitem[Cai et~al.(2023)Cai, Cardoso, Palomar, and Ying]{CaiCardosoPalomarYing2023}
J.-F. Cai, J.~V. d.~M. Cardoso, D.~P. Palomar, and J.~Ying.
\newblock Fast projected {N}ewton-like method for precision matrix estimation under total positivity.
\newblock In \emph{Advances in Neural Information Processing Systems}, volume~36. Curran Associates, Inc., 2023.

\bibitem[Calandriello et~al.(2018)Calandriello, Lazaric, Koutis, and Valko]{CalandrielloLazaricKoutisValko2018}
D.~Calandriello, A.~Lazaric, I.~Koutis, and M.~Valko.
\newblock Improved large-scale graph learning through ridge spectral sparsification.
\newblock In \emph{Proceedings of the 35th International Conference on Machine Learning}, volume~80 of \emph{Proceedings of Machine Learning Research}, pages 688--697. PMLR, 2018.

\bibitem[Cheng et~al.(2015)Cheng, Cheng, Liu, Peng, and Teng]{ChengChengLiuPengTeng2015}
D.~Cheng, Y.~Cheng, Y.~Liu, R.~Peng, and S.-H. Teng.
\newblock Efficient sampling for {G}aussian graphical models via spectral sparsification.
\newblock In \emph{Proceedings of the 28th Conference on Learning Theory}, volume~40 of \emph{Proceedings of Machine Learning Research}, pages 364--390. PMLR, 2015.

\bibitem[Fallat et~al.(2017)Fallat, Lauritzen, Sadeghi, Uhler, Wermuth, and Zwiernik]{Fallat2017}
S.~Fallat, S.~L. Lauritzen, K.~Sadeghi, C.~Uhler, N.~Wermuth, and P.~Zwiernik.
\newblock Total positivity in {M}arkov structures.
\newblock \emph{The Annals of Statistics}, 45\penalty0 (3):\penalty0 1152--1184, 2017.

\bibitem[Friedman et~al.(2008)Friedman, Hastie, and Tibshirani]{FHT2007}
J.~Friedman, T.~Hastie, and R.~Tibshirani.
\newblock Sparse inverse covariance estimation with the graphical lasso.
\newblock \emph{Biostatistics}, 9\penalty0 (3):\penalty0 432--441, 2008.

\bibitem[Fung et~al.(2011)Fung, Hariharan, Harvey, and Panigrahi]{FungHKPR2011}
W.~S. Fung, R.~Hariharan, N.~J.~A. Harvey, and D.~Panigrahi.
\newblock A general framework for graph sparsification.
\newblock \emph{Proceedings of the 43rd Annual {ACM} Symposium on Theory of Computing}, pages 71--80, 2011.

\bibitem[Koutis et~al.(2012)Koutis, Levin, and Peng]{KLP12}
I.~Koutis, A.~Levin, and R.~Peng.
\newblock Improved spectral sparsification and numerical algorithms for {SDD} matrices.
\newblock In \emph{Proceedings of the 29th International Symposium on Theoretical Aspects of Computer Science ({STACS}'12)}, pages 266--277, 2012.

\bibitem[Kyng and Sachdeva(2016)]{KyngSachdeva2016}
R.~Kyng and S.~Sachdeva.
\newblock Approximate {G}aussian elimination for {L}aplacians: Fast, sparse, and simple.
\newblock In \emph{2016 {IEEE} 57th Annual Symposium on Foundations of Computer Science ({FOCS})}, pages 573--582. IEEE, 2016.

\bibitem[Lauritzen(1996)]{lauritzen1996}
S.~L. Lauritzen.
\newblock \emph{Graphical Models}, volume~17 of \emph{Oxford Statistical Science Series}.
\newblock Clarendon Press, Oxford, 1996.

\bibitem[Lauritzen and Zwiernik(2022)]{LZ2022}
S.~L. Lauritzen and P.~Zwiernik.
\newblock Locally associated graphical models and mixed convex exponential families.
\newblock \emph{The Annals of Statistics}, 50\penalty0 (5):\penalty0 3009--3038, 2022.

\bibitem[Lauritzen et~al.(2019)Lauritzen, Uhler, and Zwiernik]{LauritzenUhlerZwiernik2019}
S.~L. Lauritzen, C.~Uhler, and P.~Zwiernik.
\newblock Maximum likelihood estimation in {G}aussian models under total positivity.
\newblock \emph{The Annals of Statistics}, 47\penalty0 (4):\penalty0 1835--1863, 2019.

\bibitem[Lauritzen et~al.(2021)Lauritzen, Uhler, and Zwiernik]{LauritzenUhlerZwiernik2021}
S.~L. Lauritzen, C.~Uhler, and P.~Zwiernik.
\newblock Total positivity in exponential families with application to binary variables.
\newblock \emph{The Annals of Statistics}, 49\penalty0 (3):\penalty0 1436--1459, 2021.

\bibitem[Sadhanala et~al.(2016)Sadhanala, Wang, and Tibshirani]{SadhanalaWangTibshirani2016}
V.~Sadhanala, Y.-X. Wang, and R.~J. Tibshirani.
\newblock Graph sparsification approaches for {L}aplacian smoothing.
\newblock In \emph{Proceedings of the 19th International Conference on Artificial Intelligence and Statistics}, volume~51 of \emph{Proceedings of Machine Learning Research}, pages 1250--1259. PMLR, 2016.

\bibitem[Slawski and Hein(2015)]{SlawskiHein2015}
M.~Slawski and M.~Hein.
\newblock Estimation of positive definite {M}-matrices and structure learning for attractive {G}aussian {M}arkov random fields.
\newblock \emph{Linear Algebra and its Applications}, 473:\penalty0 145--179, 2015.

\bibitem[Spielman and Srivastava(2011)]{SpSr2011}
D.~A. Spielman and N.~Srivastava.
\newblock Graph sparsification by effective resistances.
\newblock \emph{SIAM Journal on Computing}, 40\penalty0 (6):\penalty0 1913--1926, 2011.

\bibitem[Spielman and Teng(2004)]{ST2004}
D.~A. Spielman and S.-H. Teng.
\newblock Nearly-linear time algorithms for graph partitioning, graph sparsification, and solving linear systems.
\newblock In \emph{Proceedings of the 36th Annual {ACM} Symposium on Theory of Computing ({STOC}'04)}, pages 81--90, 2004.

\bibitem[Wang et~al.(2023)Wang, Ying, and Palomar]{WangYingPalomar2023}
X.~Wang, J.~Ying, and D.~P. Palomar.
\newblock Learning large-scale {MTP}$_2$ {G}aussian graphical models via bridge-block decomposition.
\newblock In \emph{Advances in Neural Information Processing Systems}, volume~36. Curran Associates, Inc., 2023.

\bibitem[Wang et~al.(2020)Wang, Roy, and Uhler]{WangRoyUhler2020}
Y.~Wang, U.~Roy, and C.~Uhler.
\newblock Learning high-dimensional {G}aussian graphical models under total positivity without adjustment of tuning parameters.
\newblock In \emph{Proceedings of the 23rd International Conference on Artificial Intelligence and Statistics ({AISTATS} 2020)}, volume 108 of \emph{Proceedings of Machine Learning Research}, pages 2698--2708. PMLR, 2020.

\bibitem[Wang et~al.(2022)Wang, Zhao, and Feng]{WangZhaoFeng2022}
Y.~Wang, Z.~Zhao, and Z.~Feng.
\newblock Scalable graph topology learning via spectral densification.
\newblock In \emph{Proceedings of the 15th {ACM} International Conference on Web Search and Data Mining ({WSDM} '22)}, pages 1099--1108. ACM, 2022.

\bibitem[Ying et~al.(2023)Ying, Cardoso, and Palomar]{YingCardosoPalomar2023}
J.~Ying, J.~V. d.~M. Cardoso, and D.~P. Palomar.
\newblock Adaptive estimation of graphical models under total positivity.
\newblock In \emph{Proceedings of the 40th International Conference on Machine Learning ({ICML} 2023)}, volume 202 of \emph{Proceedings of Machine Learning Research}. PMLR, 2023.

\end{thebibliography}

\appendix

\section{BSS algorithm: mechanics and complexity}\label{sec:bss_details}

We outline the deterministic linear-size sparsifier of \citet{BatsonSpielmanSrivastava2012} referenced in \Cref{thm:bss}, both to fix notation used in our pipeline (\Cref{alg:bss_refit}) and to give the per-iteration cost analysis underlying the runtime claim.

Any weighted-graph Laplacian admits the rank-one decomposition $L=\sum_{e\in E}c_e\,b_eb_e^\top$, where $c_e>0$ is the edge weight from \Cref{sec:graph_prelims} (a \emph{conductance} in the resistor-network reading) and $b_e:=e_i-e_j\in\RR^d$ is the signed indicator of the edge $e=\{i,j\}$. On $\one^\perp$, write the spectral decomposition $L=Q\Lambda Q^\top$ with $Q\in\RR^{d\times(d-1)}$ collecting the eigenvectors associated with the nonzero eigenvalues and $\Lambda\succ 0$ the corresponding diagonal of eigenvalues, and set $v_e:=\sqrt{c_e}\,\Lambda^{-1/2}Q^\top b_e\in\RR^{d-1}$. The decomposition then becomes $\sum_{e\in E}v_ev_e^\top=I_{d-1}$, so building a spectral approximation of $L$ is equivalent to choosing nonnegative weights $s_e\ge 0$ supported on a small subset of edges such that the partial sum $A:=\sum_{e\in E}s_e v_ev_e^\top$ approximates $I_{d-1}$ in Loewner order. Pulling back through $Q$ and $\Lambda^{1/2}$ recovers a sparse Laplacian.

The BSS algorithm builds $A$ greedily over $q:=\lceil\eta(d-1)\rceil$ iterations, governed by two scalar barriers $\ell<\lambda_{\min}(A)\le\lambda_{\max}(A)<u$ together with the \emph{barrier potentials}
\[
\Phi_u(A,u):=\tr(uI_{d-1}-A)^{-1},\qquad \Phi_\ell(A,\ell):=\tr(A-\ell I_{d-1})^{-1},
\]
each of which diverges as the spectrum of $A$ approaches the corresponding barrier. At every step both barriers are nudged outward by fixed positive increments $\delta_U,\delta_L$ that depend only on $\eta$, and a single edge $e$ together with a weight $\alpha>0$ is selected so that the rank-one update $A\mapsto A+\alpha v_ev_e^\top$ leaves both potentials nonincreasing under the corresponding shifts of the barriers. \citet{BatsonSpielmanSrivastava2012} prove that such a feasible pair $(e,\alpha)$ always exists and that, after $q$ steps, the eigenvalues of $A$ lie in an interval $[\ell_q,u_q]$ with ratio $u_q/\ell_q=\kappa(\eta)$. Pulling back to the original coordinates produces a Laplacian $\widetilde L$ supported on at most $q$ edges and satisfying the Loewner bound of \Cref{thm:bss}.

Computationally, the initial eigendecomposition costs $O(d^3)$, and there are $q=\lceil\eta(d-1)\rceil$ iterations. Each iteration must (i) evaluate the barrier potentials, which depend on the inverses of $(uI-A)$ and $(A-\ell I)$, and (ii) compute the candidate-score reduction $\alpha_e\propto v_e^\top M v_e$ for $M\in\{(uI-A)^{-1},(A-\ell I)^{-1}\}$, evaluated at all $m$ candidate edges via the matrix product $U:=MV$ with $V\in\RR^{(d-1)\times m}$. With direct (re-)inversion at each step, (i) costs $O(d^3)$ per iteration and (ii) costs $O(d^2 m)$ per iteration, giving a total runtime of $O(\eta(d^4+d^3 m))$. Maintaining both $M$ and the precomputed score matrix $U$ incrementally via Sherman--Morrison rank-one updates after each accepted edge reduces these to $O(d^2)$ and $O(d m)$ respectively per iteration, yielding the asymptotically faster runtime of $O(\eta(d^3+d^2 m))$ established by \citet{BatsonSpielmanSrivastava2012}. Our reference R implementation uses direct inversions for simplicity.

\section{Proofs}\label{sec:proofs}

\subsection{Proofs for \Cref{sec:gaussian}}

\begin{proof}[Proof of \Cref{thm:positive_row_sum_scaling}]
Consider the function $\Psi(\xi):=\tfrac12\xi^\top W\xi-\sum_{i=1}^d u_i\log\xi_i$ on $\RR_{>0}^d$. Since $W\succ 0$, $\Psi$ is strictly convex and coercive: the quadratic term diverges as $\|\xi\|\to\infty$, and $-u_i\log\xi_i\to+\infty$ if $\xi_i\downarrow 0$. Hence $\Psi$ has a unique minimizer $\xi^*\in\RR_{>0}^d$, whose first-order conditions give $\xi_i^*(W\xi^*)_i=u_i$ for each $i$, that is, $\diag(\xi^*)W\xi^*=u$. Setting $\Xi:=\diag(\xi^*)$ yields $\Xi W\Xi\one=u$.
\end{proof}

\begin{remark}[Computing the scaling $\Xi$]\label{rem:sinkhorn_iteration}
The positive diagonal $\Xi$ asserted by \Cref{thm:positive_row_sum_scaling} can be computed in practice as follows. When the estimated precision matrix already satisfies $\widehat K\one>0$, one may take $\Xi=I_d$ and $u=\widehat K\one$, so no scaling is needed. Otherwise, we fix the target $u=\one_d$ and solve $\Xi\widehat K\Xi\one=\one$ by the coordinatewise fixed-point iteration
\[
\xi_i \;\leftarrow\; \frac{-r_i+\sqrt{r_i^{\,2}+4\widehat K_{ii}}}{2\widehat K_{ii}},
\qquad r_i := \sum_{j\neq i}\widehat K_{ij}\,\xi_j,
\]
which is the diagonal update of the Sinkhorn-type problem associated with \Cref{thm:positive_row_sum_scaling} and converges linearly from any strictly positive initialization. Its per-iteration cost is $O(d^2)$, negligible compared to the cost of the BSS sparsification step in \Cref{alg:bss_refit}.
\end{remark}

\begin{proposition}\label{prop:B_is_SDDM}
With $\Xi$ and $u$ as above, $B:=\Xi\widehat K\Xi$ is a symmetric positive definite $M$-matrix satisfying $B\one=u>0$. Equivalently, $B$ is an SDDM matrix.
\end{proposition}

\begin{proof}[Proof of \Cref{prop:B_is_SDDM}]
Because $\Xi$ is positive diagonal and $\widehat K$ is symmetric positive definite, $B=\Xi\widehat K\Xi$ is symmetric positive definite. Moreover, $B_{ij}=\xi_i\xi_j\widehat K_{ij}\le 0$ for $i\neq j$, and by construction $B\one=u>0$. For each $i$,
\[
B_{ii}-\sum_{j\neq i}|B_{ij}|
= B_{ii}+\sum_{j\neq i}B_{ij}
= u_i > 0,
\]
so $B$ is strictly diagonally dominant and hence an SDDM matrix.
\end{proof}

\begin{proposition}\label{prop:LB_is_laplacian}
With $B$ as in \Cref{prop:B_is_SDDM} and $D:=\diag(u)$, the matrix $L_B:=B-D$ is a graph Laplacian, and its edge set coincides with the sparsity pattern of $\widehat K$. In particular, if $\widehat K$ has a connected sparsity pattern, then $L_B$ is the Laplacian of a connected weighted graph.
\end{proposition}

\begin{proof}[Proof of \Cref{prop:LB_is_laplacian}]
Since $B_{ij}\le 0$ for $i\neq j$, we have $(L_B)_{ij}=B_{ij}\le 0$ for $i\neq j$. Moreover, $L_B$ is symmetric and $L_B\one = B\one - D\one = u-u=\zero$. Hence $L_B$ is a graph Laplacian. The edge set of $L_B$ consists of all pairs $\{i,j\}$ with $B_{ij}<0$, which coincides with the edge set of $\widehat K$ since $B_{ij}=\xi_i\xi_j\widehat K_{ij}$ and $\xi_i>0$. In particular, if $\widehat K$ has a connected sparsity pattern, then $L_B$ is the Laplacian of a connected graph.
\end{proof}

\begin{proposition}\label{prop:decomp_approx}
Let $\widetilde L_B$ be any Laplacian with $(1-\varepsilon)L_B\NSD\widetilde L_B\NSD(1+\varepsilon)L_B$ on $\one^\perp$, and set $\widetilde B:=\widetilde L_B+D$. Then $\widetilde B$ is SDDM and $(1-\varepsilon)B\NSD\widetilde B\NSD(1+\varepsilon)B$ on $\RR^d$.
\end{proposition}

\begin{proof}[Proof of \Cref{prop:decomp_approx}]
For the upper bound,
\[
(1+\varepsilon)B - \widetilde B
= \bigl((1+\varepsilon)L_B - \widetilde L_B\bigr) + \varepsilon D.
\]
The spectral sparsification guarantee gives $(1+\varepsilon)L_B - \widetilde L_B \PSD 0$ on $\one^\perp$. Since both $L_B$ and $\widetilde L_B$ are graph Laplacians, the vector $\one$ is in their kernel, so the inequality extends to all of $\RR^d$. The second term $\varepsilon D\succ 0$ since $D=\diag(u)\succ 0$. Hence $(1+\varepsilon)B-\widetilde B\succ 0$. The lower bound follows analogously:
\[
\widetilde B - (1-\varepsilon)B
= \bigl(\widetilde L_B - (1-\varepsilon)L_B\bigr) + \varepsilon D
\PSD 0 + \varepsilon D \succ 0.
\]
In particular, $\widetilde B\succ 0$. Since $\widetilde L_B$ is a graph Laplacian, it has nonpositive off-diagonal entries, and $D$ is diagonal, so $\widetilde B_{ij}=(\widetilde L_B)_{ij}\le 0$ for $i\neq j$. Hence $\widetilde B$ is a symmetric positive definite $M$-matrix. Finally, $\widetilde B\one = \widetilde L_B\one + D\one = \zero + u = u >0$, so $\widetilde B$ is SDDM.
\end{proof}

\begin{proof}[Proof of \Cref{thm:gaussian_pullback}]
Symmetry and positive definiteness are preserved under positive diagonal congruence (the map $A\mapsto DAD$ with $D$ a positive diagonal matrix; here $D=\Xi^{-1}$). For $i\neq j$, $(\widetilde K)_{ij}=\xi_i^{-1}\xi_j^{-1}\widetilde B_{ij}\le 0$, so $\widetilde K$ is a symmetric positive definite $M$-matrix.

Now let $y\in\RR^d$ and set $z:=\Xi^{-1}y$. Then $y^\top \widehat K y = z^\top Bz$ and $y^\top \widetilde K y = z^\top \widetilde Bz$. Using \Cref{prop:decomp_approx},
\[
(1-\varepsilon) y^\top \widehat K y
\le y^\top \widetilde K y
\le (1+\varepsilon)y^\top \widehat K y
\qquad \forall y\in\RR^d,
\]
which is the stated Loewner-order bound. Finally, since the scaling factors $\xi_i$ are nonzero, $(\widetilde K)_{ij}=0$ if and only if $\widetilde B_{ij}=0$.
\end{proof}

\begin{proof}[Proof of Corollary~\ref{cor:mtp2_preserved_kl}]
The first claim holds because $\widetilde K$ is a symmetric positive definite $M$-matrix (\Cref{thm:gaussian_pullback}), which is the \MTPtwo{} condition for multivariate Gaussians. For the KL bound, recall that for zero-mean Gaussians
\[
2D_{\mathrm{KL}}\!\bigl(\mathcal{N}(0,\widehat K^{-1})\,\big\|\,\mathcal{N}(0,\widetilde K^{-1})\bigr)
=
\tr(\widetilde K\widehat K^{-1}) - d + \log\det(\widehat K\widetilde K^{-1})
=
\sum_{i=1}^d(\mu_i - 1 - \log\mu_i),
\]
where $\mu_1,\ldots,\mu_d$ are the eigenvalues of $\widetilde K\widehat K^{-1}$. Since $(1-\varepsilon)\widehat K\NSD\widetilde K\NSD(1+\varepsilon)\widehat K$, each $\mu_i\in[1-\varepsilon,1+\varepsilon]$. Writing $\mu_i=1+t_i$ with $|t_i|\le\varepsilon\le\tfrac12$ and using $t-\log(1+t)\le t^2$ for $|t|\le\tfrac12$, we get $\mu_i-1-\log\mu_i\le\varepsilon^2$, and summing over $i$ gives $2D_{\mathrm{KL}}\le d\varepsilon^2$.
\end{proof}

\begin{proof}[Proof of \Cref{thm:gaussian_loglik}]
For zero-mean Gaussians the Kullback--Leibler divergence has the closed form
\[
2D_{\mathrm{KL}}\!\bigl(\mathcal{N}(0,\widehat K^{-1})\,\big\|\,\mathcal{N}(0,\widetilde K^{-1})\bigr)
\;=\;
\tr(\widetilde K\widehat K^{-1}) - d - \log\det(\widetilde K\widehat K^{-1}),
\]
so $\log\det\widetilde K - \log\det\widehat K = \tr(\widetilde K\widehat K^{-1}) - d - 2D_{\mathrm{KL}}$. Substituting into $\ell(\widetilde K;T) - \ell(\widehat K;T) = (\log\det\widetilde K - \log\det\widehat K) - \tr((\widetilde K - \widehat K)T)$ and using $\tr(\widetilde K\widehat K^{-1}) - d = \tr((\widetilde K-\widehat K)\widehat K^{-1})$ gives
\[
\ell(\widetilde K;T) - \ell(\widehat K;T)
\;=\;
-2D_{\mathrm{KL}} + \tr\!\bigl((\widetilde K-\widehat K)(\widehat K^{-1}-T)\bigr)
\;=\;
-2D_{\mathrm{KL}} - \tr\!\bigl((\widetilde K-\widehat K)(T-\widehat K^{-1})\bigr),
\]
which is the identity~\eqref{eq:loglik_identity}.

For the bound~\eqref{eq:loglik_bound}, \Cref{cor:mtp2_preserved_kl} gives $0\le 2D_{\mathrm{KL}}\le d\varepsilon^2$. The Loewner sandwich $-\varepsilon\widehat K\NSD\widetilde K-\widehat K\NSD\varepsilon\widehat K$ is equivalent to $\|\widehat K^{-1/2}(\widetilde K-\widehat K)\widehat K^{-1/2}\|_{\mathrm{op}}\le\varepsilon$. Writing
\[
\tr\!\bigl((\widetilde K-\widehat K)(T-\widehat K^{-1})\bigr)
\;=\;
\tr\!\bigl(\widehat K^{-1/2}(\widetilde K-\widehat K)\widehat K^{-1/2}\,\cdot\,\widehat K^{1/2}(T-\widehat K^{-1})\widehat K^{1/2}\bigr)
\]
and applying the operator--trace H\"older inequality $|\tr(M_1 M_2)|\le\|M_1\|_{\mathrm{op}}\|M_2\|_*$,
\[
\bigl|\tr\!\bigl((\widetilde K-\widehat K)(T-\widehat K^{-1})\bigr)\bigr|
\;\le\;\varepsilon\,\bigl\|\widehat K^{1/2}(T-\widehat K^{-1})\widehat K^{1/2}\bigr\|_*
\;=\;\varepsilon\,R(T).
\]
Combining the two contributions gives $|\ell(\widetilde K;T)-\ell(\widehat K;T)|\le d\varepsilon^2+\varepsilon R(T)$.
\end{proof}

\begin{proof}[Proof of Corollary~\ref{cor:loglik_refit}]
Define the Bregman divergence associated with the strictly convex generator $f(K):=-\log\det K$ on the cone of positive definite matrices:
\begin{align*}
B_f(K_1,K_2) &:= f(K_1)-f(K_2)-\langle\nabla f(K_2),K_1-K_2\rangle \\
&\phantom{:}= \log\det K_2-\log\det K_1+\tr(K_2^{-1}(K_1-K_2)).
\end{align*}
A direct computation gives the closed-form identity
\[
B_f(K_1,K_2) = 2\,D_{\mathrm{KL}}\!\bigl(\mathcal{N}(0,K_2^{-1})\,\big\|\,\mathcal{N}(0,K_1^{-1})\bigr).
\]
For any positive semidefinite $S$, the negative log-likelihood $-\ell(K;S) = -\log\det K + \tr(KS) + \mathrm{const}$ differs from $f(K)$ only by a linear term in $K$, which contributes nothing to the Bregman divergence. Hence the Bregman divergence of $-\ell(\cdot;S)$ coincides with $B_f$, and we have the exact Taylor identity, valid for every $K,K'\succ 0$:
\begin{equation}\label{eq:taylor_loglik}
-\ell(K;S) \;=\; -\ell(K';S) + \langle -\nabla_K\ell(K';S),\,K-K'\rangle + B_f(K,K').
\end{equation}

The feasible set $C := \{K\succ 0:\,K\text{ is \MTPtwo{}},\;\supp(K)\subseteq E_{\mathrm{sp}}\}$ is convex, being the intersection of the PD cone, the linear subspace $\{K:K_{ij}=0\text{ for all }i\ne j\text{ with }\{i,j\}\notin E_{\mathrm{sp}}\}$, and the convex cone of nonpositive off-diagonals. Since $\widetilde K^{\mathrm{MLE}}$ maximizes $\ell(\cdot;S)$ over $C$, the first-order optimality condition reads
\[
\langle -\nabla_K\ell(\widetilde K^{\mathrm{MLE}};S),\,K-\widetilde K^{\mathrm{MLE}}\rangle \;\ge\; 0 \qquad\text{for every }K\in C.
\]
Specializing~\eqref{eq:taylor_loglik} to $K'=\widetilde K^{\mathrm{MLE}}$ and combining with this variational inequality yields, for every $K\in C$,
\[
\ell(\widetilde K^{\mathrm{MLE}};S) - \ell(K;S) \;\ge\; B_f(K,\widetilde K^{\mathrm{MLE}}) \;=\; 2\,D_{\mathrm{KL}}\!\bigl(\mathcal{N}(0,(\widetilde K^{\mathrm{MLE}})^{-1})\,\big\|\,\mathcal{N}(0,K^{-1})\bigr).
\]
The matrix $\widetilde K$ is \MTPtwo{} and supported on $E_{\mathrm{sp}}$ by \Cref{thm:gaussian_pullback}, so $\widetilde K\in C$, and applying the inequality with $K=\widetilde K$ proves~\eqref{eq:refit_pythagoras}. Chaining with the lower bound $\ell(\widetilde K;S)-\ell(\widehat K;S)\ge-d\varepsilon^2-\varepsilon R(S)$ from \Cref{thm:gaussian_loglik} yields the second claim.
\end{proof}

\begin{proposition}[Non-asymptotic Frobenius bound and consistency]\label{prop:gaussian_consistency}
Let $\widehat K\in\RR^{d\times d}$ be an \MTPtwo{} positive definite $M$-matrix and $\widetilde K$ the output of \Cref{thm:gaussian_pullback} applied to $\widehat K$ at level $\varepsilon\in(0,1)$. Then $\widetilde K$ is itself an \MTPtwo{} positive definite $M$-matrix, and for every $K^\star\in\RR^{d\times d}$,
\begin{equation}\label{eq:bias_decomp_app}
\|\widetilde K - K^\star\|_F
\;\le\;
\underbrace{\varepsilon\sqrt d\,\|\widehat K\|_2}_{\text{sparsification bias}}
\;+\;
\underbrace{\|\widehat K - K^\star\|_F}_{\text{starting-estimator error}}.
\end{equation}
The first term is a deterministic, $\varepsilon$-controlled bias from sparsification; the second is whatever error $\widehat K$ already carries.
Suppose now that $K^\star$ is the true precision matrix of a Gaussian \MTPtwo{} model with $d$ fixed, and that $\widehat K_n$ is a consistent estimator with $\|\widehat K_n - K^\star\|_F = O_p(r_n)$ for some sequence $r_n\to 0$. Applying~\eqref{eq:bias_decomp_app} with $\varepsilon=\varepsilon_n\in(0,1)$ gives
\[
\|\widetilde K_n - K^\star\|_F = O_p\bigl(r_n + \varepsilon_n\sqrt d\bigr).
\]
In particular, $\widetilde K_n$ is consistent whenever $\varepsilon_n\to 0$, with rate $\max(r_n,\varepsilon_n\sqrt d)$, and the rate $r_n$ is preserved when $\varepsilon_n = O(r_n/\sqrt d)$. The unpenalized Gaussian \MTPtwo{}-MLE \citep{LauritzenUhlerZwiernik2019} is a valid choice of $\widehat K_n$.
\end{proposition}

\begin{proof}[Proof of \Cref{prop:gaussian_consistency}]
That $\widetilde K$ is an \MTPtwo{} positive definite $M$-matrix follows from \Cref{thm:gaussian_pullback} and \Cref{cor:mtp2_preserved_kl}. From the Loewner bound $(1-\varepsilon)\widehat K \NSD \widetilde K \NSD (1+\varepsilon)\widehat K$ in \Cref{thm:gaussian_pullback}, the eigenvalues of $\widetilde K - \widehat K$ all lie in $[-\varepsilon\lambda_{\max}(\widehat K),\,\varepsilon\lambda_{\max}(\widehat K)]$, so
\[
\|\widetilde K - \widehat K\|_F
\;\le\; \sqrt{d}\,\|\widetilde K - \widehat K\|_2
\;\le\; \varepsilon\sqrt{d}\,\|\widehat K\|_2.
\]
Combining with the triangle inequality $\|\widetilde K - K^\star\|_F \le \|\widetilde K - \widehat K\|_F + \|\widehat K - K^\star\|_F$ proves~\eqref{eq:bias_decomp_app}. For the asymptotic part, $\|\cdot\|_2\le\|\cdot\|_F$ together with the triangle inequality gives $\|\widehat K_n\|_2 \le \|\widehat K_n - K^\star\|_F + \|K^\star\|_2 = O_p(r_n) + O(1) = O_p(1)$, and substitution into~\eqref{eq:bias_decomp_app} yields $\|\widetilde K_n - K^\star\|_F \le \varepsilon_n\sqrt d\,O_p(1) + O_p(r_n) = O_p(r_n + \varepsilon_n\sqrt d)$. The two consistency regimes follow.
\end{proof}

\section{Training and out-of-sample regimes}\label{sec:regimes}

The asymmetric two-sided bound \eqref{eq:loglik_bound} of \Cref{thm:gaussian_loglik} admits a sharper interpretation in the two natural regimes of practical interest, depending on whether the evaluation covariance $S$ coincides with the training covariance used to fit $\widehat K$ or is independent of it. We make these two cases explicit below.

\textbf{Training regime.} Suppose we evaluate the log-likelihood at the same training covariance $S=S_{\mathrm{train}}$ that was used to fit $\widehat K$, and that $\widehat K$ is the \MTPtwo{}-constrained MLE on $S_{\mathrm{train}}$, that is, the minimizer of $-\log\det K+\tr(KS_{\mathrm{train}})$ subject to $K_{ij}\le 0$ for all $i\ne j$. Stationarity and complementary slackness for this constrained problem read
\begin{equation}\label{eq:kkt_mtp2}
\widehat K^{-1} - S_{\mathrm{train}} - \Lambda \;=\; 0,
\qquad
\Lambda_{ij}\,\widehat K_{ij} = 0\ \text{ for all } i\ne j,
\end{equation}
where $\Lambda$ is a symmetric matrix of multipliers with $\Lambda_{ij}\ge 0$ off the diagonal (associated with $K_{ij}\le 0$) and $\Lambda_{ii}=0$. Complementary slackness means $\Lambda$ is supported only on the non-edges of $\widehat K$.

Substituting $S_{\mathrm{train}}-\widehat K^{-1}=-\Lambda$ into the Bregman identity \eqref{eq:loglik_identity} gives
\[
\ell(\widetilde K;S_{\mathrm{train}}) - \ell(\widehat K;S_{\mathrm{train}})
\;=\;
-2\,D_{\mathrm{KL}}\!\bigl(\mathcal{N}(0,\widehat K^{-1})\,\big\|\,\mathcal{N}(0,\widetilde K^{-1})\bigr)
\;+\;\tr\!\bigl((\widetilde K-\widehat K)\Lambda\bigr).
\]
By complementary slackness, $\tr(\widehat K\Lambda)=\sum_{i\ne j}\widehat K_{ij}\Lambda_{ij}=0$, so the trace term simplifies to $\tr(\widetilde K\Lambda)=\sum_{i\ne j}\widetilde K_{ij}\Lambda_{ij}$. Since $\widetilde K$ is \MTPtwo{} by \Cref{cor:mtp2_preserved_kl}, $\widetilde K_{ij}\le 0$ for $i\ne j$, while $\Lambda_{ij}\ge 0$; hence $\tr(\widetilde K\Lambda)\le 0$. Both terms on the right-hand side are therefore nonpositive, so
\[
\ell(\widetilde K;S_{\mathrm{train}}) \;\le\; \ell(\widehat K;S_{\mathrm{train}}),
\]
with equality iff $\widetilde K=\widehat K$. On training data the sparsifier therefore cannot strictly improve on the \MTPtwo{}-MLE. When $\varepsilon\le\tfrac12$, the lower bound \eqref{eq:loglik_bound} additionally caps the loss at $d\varepsilon^2$ nats.

\textbf{Out-of-sample regime.} Now suppose $S$ is unrelated to $\widehat K$. The natural example is a held-out test covariance $S=S_{\mathrm{test}}$ computed from a sample independent of the training set. Here the residual $S-\widehat K^{-1}$ is generically nonzero, $R(S)>0$, and the upper bound \eqref{eq:loglik_upper} now permits a strict gain
\[
\ell(\widetilde K;S) \;>\; \ell(\widehat K;S).
\]
By the identity \eqref{eq:loglik_identity}, this gain materializes precisely when
\[
\tr\!\bigl((\widetilde K-\widehat K)(S-\widehat K^{-1})\bigr) \;<\; -2D_{\mathrm{KL}}(\widehat K\,\|\,\widetilde K),
\]
that is, when the perturbation $\widetilde K-\widehat K$ produced by sparsification aligns with the train--test discrepancy $S-\widehat K^{-1}$ strongly enough to overcome the KL cost of replacing $\widehat K$ by $\widetilde K$. The natural setting in which this happens is when $\widehat K$ has overfitted training noise into spurious edges that BSS subsequently drops. In that case the gain in predictive fit on test data outweighs the structural KL loss, in line with the standard bias--variance picture.

\end{document}